\documentclass[a4paper,UKenglish,cleveref, autoref, thm-restate]{lipics-v2021}

\usepackage{algorithm}
\usepackage{algpseudocode}

\newcommand{\abs}[1]{\left|#1\right|}
\newcommand{\paren}[1]{\left(#1\right)}
\newcommand{\set}[1]{\left\{#1\right\}}

\nolinenumbers

\DeclareMathOperator{\sat}{sat}

\newcommand{\E}{\mathbb{E}}

\newtheorem{problem}{Problem}

\title{Distributed weak independent sets in hypergraphs: Upper and lower bounds} %TODO Please add

\author{Duncan Adamson}{School of Computer Science, University of St Andrews, UK}{duncan.adamson@st-andrews.ac.uk}{0000-0003-3343-2435}{}

\author{Will Rosenbaum}{Department of Computer Science, University of Liverpool, UK}{w.rosenbaum@liverpool.ac.uk}{0000-0002-7723-9090}{}

\author{Paul G. Spirakis}{Department of Computer Science, University of Liverpool, UK}{P.Spirakis@liverpool.ac.uk}{0000-0001-5396-3749}{}

\authorrunning{Duncan Adamson, Will Rosenbaum, and Paul G. Spirakis} %TODO mandatory. First: Use abbreviated first/middle names. Second (only in severe cases): Use first author plus 'et al.'

\Copyright{Duncan Adamson, Will Rosenbaum, and Paul G. Spirakis} %TODO mandatory, please use full first names. LIPIcs license is "CC-BY";  http://creativecommons.org/licenses/by/3.0/

 %TODO mandatory: Please choose ACM 2012 classifications from https://dl.acm.org/ccs/ccs_flat.cfm 

 %TODO mandatory; please add comma-separated list of keywords

\category{} %optional, e.g. invited paper

\relatedversion{} %optional, e.g. full version hosted on arXiv, HAL, or other respository/website
\acknowledgements{}%optional

\EventEditors{John Q. Open and Joan R. Access}
\EventNoEds{2}
\EventLongTitle{42nd Conference on Very Important Topics (CVIT 2016)}
\EventShortTitle{CVIT 2016}
\EventAcronym{CVIT}
\EventYear{2016}
\EventDate{December 24--27, 2016}
\EventLocation{Little Whinging, United Kingdom}
\EventLogo{}
\SeriesVolume{42}
\ArticleNo{23}
\begin{document}

\maketitle

%TODO mandatory: add short abstract of the document
\begin{abstract}
In this paper, we consider the problem of finding weak independent sets in a distributed network represented by a hypergraph. In this setting, each edge contains a set of  $r$ vertices rather than simply a pair, as in a standard graph. A $k$-weak independent set in a hypergraph is a set where no edge contains more than $k$ vertices in the independent set. We focus two variations of this problem. First, we study the problem of finding $k$-weak maximal independent sets, $k$-weak independent sets where each vertex belongs to at least one edge with $k$ vertices in the independent set. Second we introduce a weaker variant that we call $(\alpha, \beta)$-independent sets where the independent set is $\beta$-weak, and each vertex belongs to at least one edge with at least $\alpha$ vertices in the independent set. Finally, we consider the problem of finding a $(2, k)$-ruling set on hypergraphs, i.e. independent sets where no vertex is a distance of more than $k$ from the nearest member of the set.

Given a hypergraph $H$ of rank $r$ and maximum degree $\Delta$, we provide a LLL formulation for finding an $(\alpha, \beta)$-independent set when $(\beta - \alpha)^2 / (\beta + \alpha) \geq 6 \log(16 r \Delta)$, an $O(\Delta r / (\beta - \alpha + 1) + \log^* n)$ round deterministic algorithm finding an $(\alpha, \beta)$-independent set, and a $O(\Delta^2(r - k) \log r + \Delta \log r \log^* r + \log^* n)$ round algorithm for finding a $k$-weak maximal independent set. Additionally, we provide zero round randomized algorithms for finding $(\alpha, \beta)$ independent sets, when $(\beta - \alpha)^2 / (\beta + \alpha) \geq 6 c \log n + 6$ for some constant $c$, and finding an $m$-weak independent set for some $m \geq r / 2k$ where $k$ is a given parameter.
% 
% algorithms finding a $(\alpha, \beta)$-IS with $\beta \geq 12 \ln (4 r \Delta)$ and $\alpha \leq r - 12 \ln(4r\Delta)$ in $O(\log^2 n)$ rounds via an LLL (Lovasz Local Lemma) formulation, an $(\alpha, \beta)$-IS with $\beta \geq 2 \alpha$ in $O(\Delta r /\alpha + \log^* n)$ rounds, and finding a $(2, k)$-ruling set in $O(\Delta(k + 2^{k - 1}r/ 3^{k - 1}) + \log^* n)$ rounds. 
%We additionally show that, given any $k'$-weak maximal independent set, we can construct for some $k < k'$ a $k$-weak maximal independent set in $O(k' \Delta + \log^* n)$ rounds.
Finally, we provide lower bounds of $\Omega(\Delta + \log^* n)$ and $\Omega(r + \log^* n)$ on the problems of finding a $k$-weak maximal independent sets for some values of $k$.
\end{abstract}

\section{Introduction}

Independent sets are a key tool in distributed networks, allowing nodes within a graph to assert themselves as leaders over their immediate neighbors so as to help with the allocation of resources. This is of particular importance in the distributed setting, where nodes likely do not have full information of the complete graph. The best known version of the independent set problem is the \emph{maximal independent set} (MIS) problem. Informally, in non-hypergraphs, an independent set is a subset of vertices such that no pair of vertices belong to an edge in the graph. Such a set is maximal if every vertex either belongs to to the set, or has some neighbor in the set. As this property is \emph{locally checkable} \cite{linial1992locality}, each vertex in the graph can determine whether it satisfies this condition by considering only the local neighborhood. In other words, once an algorithm to find an MIS has terminated, each vertex can determine if the algorithm was locally successful by checking if it belongs to the independent set, and if any of its neighbors do.\looseness=-1

Due to the fundament nature of the problem, finding independent sets in distributed graphs is one of the most heavily studied problems in distributed computing. Starting with the classic algorithm due to Linial \cite{linial1987distributive,linial1992locality}, there has been a successive series of improvements in terms of both deterministic \cite{awerbuch1989network,barenboim2016locality,panconesi2001some,barenboim2009distributed} and randomized \cite{kutten2014distributed,kuhn2018efficient,ghaffari2016improved,ghaffari2021improved} algorithms. At the same time, there has been a significant body of work determining the lower bound on the number of rounds required for distributed algorithms in general \cite{brandt2016lower,brandt2019automatic}, and finding a MIS in particular \cite{balliu2022distributed,balliu2021improved,balliu2021lower}. As it stands, the current state-of-the-art for solving MIS deterministically, in terms of the maximum degree of the graph, matches the lower bound, giving an $O(\Delta + \log^* n)$ round algorithm for finding an MIS in an $n$ node graph, with a maximum degree of $\Delta$.\looseness=-1

There are various ways of generalizing the MIS problem to the hypergraph setting. The best studied is the \emph{weak-MIS problem}. A set $S$ is a \emph{weak-MIS} if no hyperedge in the graph is a subset of $S$. This is the primary version studied in the distributed setting, with both positive, algorithmic results \cite{kutten2014distributed} and lower bounds \cite{Balliu2023-distributed}. In this paper, we consider two generalization of the independent set problem to hypergraphs, $k$-weak independent sets, and $(\alpha, \beta)$-independent sets. A $k$-weak independent set is a set $S$ such that the intersection between any edge in the graph and $S$ has size at most $k$. Such a set is maximal if every vertex in the hypergraph either belongs to $S$, or belongs to some edge containing $k$ other vertices in $S$. An $(\alpha, \beta)$-independent set is a $\beta$-weak independent set $S$, where each vertex is either in $S$, or belongs to at least one edge containing $\alpha$ vertices in $S$. Note that a $(k, k)$-independent set is thus a $k$-weak MIS. We note that both definitions are locally checkable, in the same manner as above.\looseness=-1

Of particular interest to our paper are the works of Kutten et al. \cite{kutten2014distributed}, Kuhn and Zheng \cite{kuhn2018efficient}, and Balliu et al. \cite{Balliu2023-distributed}, all of which analyze the MIS problem on hypergraphs. In \cite{kutten2014distributed}, Kutten et al provide a $O(\log^2 n)$ round randomized algorithm for solving maximal independent sets within hypergraphs containing $n$ nodes. This is done by way of a network decomposition, partitioning the hypergraph into a collection of low diameter components, i.e. components for which the distance between any pair of nodes is minimized, in this case $O(\log n)$. Once partitioned, each component may ``centralize'' the topology of the local neighborhood into a single node, that can then solve the problem, broadcasting the solution to the other nodes in its component. Building upon this, Kuhn and Zheng\cite{kuhn2018efficient} provide an $O(\log n)$ round algorithm for finding an MIS in a \emph{linear} hypergraph, a hypergraph where no pair of edges share more than a single common vertex. They further introduce and provide an $O(\log^2 n)$ round algorithm for the \emph{generalized MIS problem}, a problem roughly equivalent to our $k$-weak MIS problem, but with a variable for each edge determining the maximum number of vertices in the edge that may belong to the set. Finally, in \cite{Balliu2023-distributed}, Balliu et al. provide a pair of deterministic algorithms for finding an MIS in a hypergraph. Given a hypergraph of rank $r$, maximum degree $\Delta$, and $n$ nodes, the authors show that an MIS can be found in both $O(\Delta^2 \log r + \Delta \log r  \log^* r + \log^* n)$ rounds, and in $O(\Delta^{O(\Delta)} \log^* r + O(\log^* n)$ rounds.\looseness=-1

We finally note two tools that can be applied directly to get naive solutions to the problems of finding a $k$-weak MIS, and an $(\alpha, \beta)$-IS. On one hand is the recent graph decomposition algorithm of Ghaffari and Grunau \cite{ghaffari2024nearoptimaldeterministicnetworkdecomposition}, allowing a deterministic graph decomposition in $\Tilde{O}(\log^2 n)$ rounds into components of diameter $O(\log n)$ using $O(\log n)$ colors, and thus either problem to be solved in $O(\log^2 n)$ rounds. On the other, there is the mention the $O(\Delta)$-coloring algorithm of Maus and Tonoyan~\cite{maus2020local}, finding such a coloring in $O(\sqrt{(\Delta) (\log \Delta )} + \log^* n)$ rounds. By finding such a coloring on the underlying graph (with maximum degree $\Delta r$), we can iterate over each color class, allowing vertices to add themselves to the set without risk of conflict in $O(\Delta r)$ rounds.\looseness=-1

\subsection{Our Contributions}

While $k$-weak MIS was (implicitly) introduced in~\cite{kuhn2018efficient}, we believe our definition of $(\alpha,\beta)$-independent set is novel. Our primary formal contributions consist of upper and lower-bounds for finding independent sets in hypergraphs in the LOCAL model. 

For lower bounds, in Section~\ref{sec:lower-bounds} we show the following:
\begin{itemize}
    \item Any algorithm for finding a $1$-weak MIS requires $\Omega(\Delta + r + \log^* n)$ rounds (Theorems~\ref{thm:one_lb} and~\ref{thm:delta-lb}).
    \item For any odd value $k$, finding a $k$-MIS in a hypergraph of rank $2k$ requires $\Omega(\Delta + \log^* n)$ rounds (Theorem~\ref{thm:delta-lb-k}).
\end{itemize}

For upper bounds, we provide a collection of deterministic and randomized algorithms solving several cases of $(\alpha, \beta)$-IS and $k$-weak MIS. In all of these results, we assume that $H$ is a hypergraph on $n$ vertices with rank $r$ and maximum degree $\Delta$, and that $\alpha$ and $\beta$ satisfy $1 \leq \alpha \leq \beta < r$.
\begin{itemize}
    \item In Section~\ref{sec:lll} we provide an LLL formulation for finding an $(\alpha, \beta)$-IS when $(\beta - \alpha)^2 / (\beta + \alpha) \geq 6 \log(16 r \Delta)$, allowing a deterministic $O(\log^2 n)$ round deterministic algorithm (Corollary~\ref{cor:lll}). Using this formulation, we show that when $(\beta - \alpha)^2 / (\beta + \alpha) \geq 6 \log n + 6$, an $(\alpha,\beta)$-IS can be found in zero rounds with high probability (Corollary~\ref{cor:zero-round}).
    \item In Section~\ref{sec:high-rank}, for hypergraphs of sufficiently high rank and low amount of intersection, we present a $0$-round algorithm producing an $m$-weak MIS with expected size of $m \geq r / 2k$ for some chosen $k$ (Lemma~\ref{lem:high-rank}).
    \item In Section~\ref{sec:alpha-beta-mis-alg}, we describe a deterministic algorithm for finding an $(\alpha, \beta)$-IS in $O(\Delta r / (\beta - \alpha + 1) + \log^* n)$ rounds (Theorem~\ref{thm:edgepartitionIS}). This result generalizes the ``trivial'' algorithm of Lemma~\ref{lem:delta-r-algorithm} for finding a $k$-weak MIS, which corresponds to this result when $k = \alpha = \beta$. 
    \item In Section~\ref{sec:k-mis} we provide a $O(\Delta^2(r - k) \log r + \Delta \log r \log^* r + \log^* n)$ round algorithm for finding a $k$-weak MIS (Theorem~\ref{thm:k-weak-mis}). This result generalizes a result of Balliu et al.~\cite{Balliu2023-distributed}, who give an algorithm with similar running time in the case that $k = r - 1$. 
    \item Finally, in Section~\ref{sec:ruling-sets} we provide an algorithm for finding a $(2, k)$-Ruling Set in $O(\Delta(k + (2^{k - 1}/3^{k - 1})r) + \log^* n)$ rounds (Corollary~\ref{alg:FindRS}).
\end{itemize}
\looseness=-1

\section{Preliminaries}

Let $H = (V, E)$ be a hypergraph consisting of the set of vertices, $V$, and edges $E$. Each edge in $E$ contains some subset of vertices from $V$. The \emph{rank} of an edge $e$ is the number of vertices in the edge. The \emph{rank} of the graph is equal to the maximum rank of any edge in $E$. The \emph{degree} of a vertex $v$ is the number of edges in $E$ containing $v$. The maximum degree of $H$ is the maximum degree of any vertex in $v$. By convention, we denote the rank of a hypergraph by $r$, and the degree by $\Delta$. Given a pair of vertices $v, u$ in $V$, the \emph{distance} between $v$ and $u$, denoted $dist(v, u)$ is the minimum number of edges needed to form a contiguous path between $v$ and $u$. For example, $dist(v, v) = 0$, $dist(v, u) = 1$ if and only if $\exists e \in E$ such that $\{v, u\} \subseteq e$.\looseness=-1

The \emph{underlying graph} of the hypergraph $H = (V, E)$ is the graph $G = (V, E')$ formed by replacing each edge $e \in E$ with a clique containing every vertex in $e$, i.e. $E' = \{(v, u) \mid \exists e \in E$ such that $(u, v) \subseteq e \}$. Given a subset of vertices $S \subseteq V$, the graph \emph{induced} by $S$ is the graph $H' = \{S, E'\}$ where $E' = \{e \cap S \mid e \in E\}$, i.e. the graph formed by removing every vertex in $V \setminus S$.

An \emph{independent set} $S$ in a hypergraph $(V, E)$ is a subset of vertices where, for any pair $u, v \in S$, where $u \neq v$, $\nexists e \in E$ such that $(u, v) \in e$. A $\psi$-\emph{coloring} $\mathbf{\Psi}$ of a hypergraph $(V, E)$ is a mapping of the set of vertices to some set of $\psi$-colors, assumed to be the set of integers $1, 2, \dots, \psi$. A $\psi$-\emph{coloring} is \emph{valid} if $\forall v, u \in V$ where $v \neq u$, either $\nexists e \in E$ such that $(v, u) \in e$ or $\mathbf{\Psi}(v) \neq \mathbf{\Psi}(u)$.\looseness=-1

\begin{definition}[$k$-weak IS]
    Given a hypergraph $G = (V, E)$ of rank $r$, a \emph{$k$-weak independent set}, denoted $k$-weak IS, is a subset $V' \subseteq V$ such that, for every $e \in E$, $\vert e \cap V' \vert \leq k$, i.e. no edge contains more than $k$ vertices in $V'$. A $k$-weak independent set $V'$ is \emph{maximal} (a $k$-weak MIS) if no super set of $V'$ is a $k$-weak independent set, i.e. each vertex not in $V'$ is adjacent to at least one edge containing $k$ members of $V'$.
\end{definition}

\begin{problem}\label{prob:weak_ind_set}
    In a given communication model, what is the minimum number of rounds needed to find a \emph{maximal $k$-weak independent set} in the graph $G$ (with high probability)?
\end{problem}

In the case that some edge $e \in E$ has rank less than or equal to $k$, it is possible that every vertex in $e$ belongs to the independent set. Further, a $1$-weak independent set on a hypergraph $H$ corresponds to the traditional definition of an independent set on the underlying graph. For convenience, in non-hypergraphs, i.e. graphs of rank $2$, we denote $1$-weak maximal independent sets as simply MIS.

We now define a generalized version of independent sets, $(\alpha, \beta)$-weak independent sets.

\begin{definition}[$(\alpha, \beta)$-IS]
    \label{def:alpha_beta_sets}
    Given a hypergraph $H = (V, E)$ of rank $r$, an \emph{$(\alpha, \beta)$-weak independent set}, denoted $(\alpha, \beta)$-IS, is a subset $V' \subseteq V$ such that, for every $e \in E$, $\vert e \cap V' \vert \leq \beta$, and every vertex $v$ belongs to at least one edge $e$ such that $\vert e \cap V' \vert \geq \alpha$.
\end{definition}

\begin{problem}
    \label{prob:a_b_is}
    Given a hypergraph $H = (V, E)$ of rank $r$, and pair $\beta, \alpha \in \mathbb{N}$ such that $\beta \geq \alpha$, what is the minimum number of rounds needed to find an $(\alpha, \beta)$-IS (with high probability)?
\end{problem}

Note that a $(\alpha, \beta)$-weak independent set is, by definition, an $\beta$-weak independent set, though not necessarily maximal. On the other hand, an $(\alpha, \alpha)$-IS is an $\alpha$-MIS.

A related by distinct problem is that of finding an \emph{$(a, b)$-ruling set}. Given a graph $G = (V, E)$, the set $S \subseteq V$ is an $(\alpha, \beta)$-ruling set if:
\begin{itemize}
    \item given any pair $v, u \in S$ where $v \neq u$, $dist(v, u) \geq a$, and,
    \item given any $v \in V$, $\exists u \in S$ such that $dist(v, u) \leq b$.
\end{itemize}

In this paper we will briefly touch upon the problem of finding a $(2, O(\log r))$-ruling set within a hypergraph of rank $r$.

\subsection{Model of computing}

In this paper, we primarily consider the LOCAL model of distributed computing. In this model, each vertex in the graph is assigned a unique ID, and is aware of the local neighborhood, i.e. the edges it belongs to and the IDs of the vertices in each edge. Each round consists of a period of computation, followed by each vertex sending some message to each neighbor. We do not place any bound on the size of the messages in this model. We note that some cited work uses the CONGEST model, a restriction of the LOCAL model where each vertex may only send messages of size $O(\log n)$. 

\subsection{Trivial Reductions}
\label{sec:trivial-algorithms}

Before presenting our main results, we outline some trivial approaches based on using existing deterministic and randomised algorithms on the underlying graph.

\begin{corollary}[3.6, \cite{ghaffari2021improved}]
\label{col:fast_randomised_lin}
    There is a randomized distributed algorithm that computes a maximal independent set in $O(\log \Delta \cdot \log \log n + \log^9 \log n)$ rounds of the CONGEST model, with high probability.
\end{corollary}

Corollary \ref{col:fast_randomised_lin} is based on finding a \emph{graph decomposition}, a decomposition of the graph into a number of small, disconnected components, allowing the problem to be solved by centralising the information of each component in some vertex which computes a valid solution, then shares this solution with the remaining vertices in the graph. We can adapt this to the hypergraph setting to get the following.

\begin{lemma}
\label{lem:fast_randomised_hyper}
    There is a randomized distributed algorithm that computes a $k$-maximal independent set in $O(\log \Delta r \cdot \log \log n + \log^9 \log n)$ rounds of the CONGEST model, with high probability for a hypergraph of rank $r$ with maximum degree $\Delta$.
\end{lemma}

In the deterministic setting, we can use the algorithm of Ghaffari and Grunau~\cite{ghaffari2024nearoptimaldeterministicnetworkdecomposition} to decompose the graph in $\Tilde{O}(\log^2 n)$ rounds into components of diameter $O(\log n)$ using $O(\log n)$ colors, i.e. a coloring of the graph with $O(\log n)$ colors such that each color class induces components of diameter at most $O(\log n)$.

\begin{lemma}
\label{lem:deterministic_decomposition}
    There is a deterministic distributed algorithm that computes a $k$-weak maximal independent set in $O(\log^2 n)$ rounds of the CONGEST model for a hypergraph of rank $r$ with maximum degree $\Delta$.
\end{lemma}

Finally, for small values of $\Delta$ and $r$, we can find an $O(r \Delta)$ coloring of the underlying graph, then iterate through each color class, adding each vertex to the independent set provided the set remains $k$-weak.

\begin{lemma}\label{lem:delta-r-algorithm}
    There is a deterministic distributed algorithm the computes a $k$-weak maximal independent set in $O(\Delta r + \log^*n)$ rounds of the LOCAL model.
\end{lemma}

\section{Lower Bounds for weak MIS}
\label{sec:lower-bounds}

Before providing our solutions to Problem \ref{prob:weak_ind_set}, we first provide a pair of lower bounds of $\Omega(r + \log^* n)$ and $\Omega(\Delta + \log^* n)$ for deterministic algorithms solving $1$-weak MIS, and a $\Omega(\Delta + \log^* n)$ lower bound for finding a $k$-MIS in a hypergraph of rank $2k$, for an odd value of $k$.\looseness=-1

% \paragraph*{An $\Omega(r + \log^* n)$ Lower Bound for $1$-weak MIS}
\paragraph*{Lower Bounds in $r$.}

\begin{theorem}\label{thm:one_lb}
    Consider the family of $r$-uniform hypergraphs with maximum degree $\Delta$ such that $r \Delta \log (r \Delta) = O(\log n)$. Then in this family, computing a $1$-weak MIS requires $\Omega(r)$ rounds in the local model.
\end{theorem}

The idea of the proof of Theorem~\ref{thm:one_lb} is to give a reduction from hypergraph maximal matching. To this end, we employ the lower bound of~\cite{Balliu2023-distributed}:

\begin{theorem}[Balliu et al.~\cite{Balliu2023-distributed}]
    \label{thm:mm_lb}
    For the family of $r$-uniform hypergraphs with maximum degree $\Delta$ and $r \Delta \log (r \Delta) = O(\log n)$, computing a maximal matching requires $\Omega(r \Delta)$ rounds in the LOCAL model.
\end{theorem}

\begin{proof}[Proof of Theorem~\ref{thm:one_lb}]
    Suppose we have an algorithm $A$ that computes a $1$-weak MIS in $T = T(\Delta, r, n)$ rounds. Given a hypergraph $G = G_0$, define a sequence of hypergraphs $G = G_0, G_1, \ldots, G_\Delta$ and sets of hyperedges $M_1, M_2, \ldots, M_\Delta$ as follows. Given $G_{i-1}$ compute $M_i$ and $G_i$ by:\looseness=-1

    \begin{enumerate}
        \item Apply $A$ to compute a $1$-weak MIS $S_i$ in $G_{i-1}$.
        \item For each vertex $v \in S_i$, choose an edge $e_v$ incident to $v$ arbitrarily and add $e_v$ to $M_i$. Note that this can be done in a single round of LOCAL given the output of $A$.
        \item Form $G_i$ by removing from $G_{i-1}$ all vertices $u$ incident to edges in $M_i$ and all edges incident to those vertices.
    \end{enumerate}

    We claim that $M = \bigcup_{i = 1}^\Delta M_i$ is a maximal matching in $G$. 
    
    First observe that each $M_i$ is a matching in $G_{i-1}$. Next, note that if $e \notin M_i$ is an edge removed in Step~3 above, then $e$ intersects some edge $e' \in M_i$, hence $e$ cannot be in any matching that contains the edges from $M_i$.\looseness=-1

    Next consider an edge $e$ in $G_{i}$---i.e., an edge that is neither in $M_i$ nor was it removed in Step~3. Since $S_i$ is a $1$-weak MIS in $G_{i-1}$, it must be the case that every vertex $v \in e$ is also contained in another edge $e'$ that intersects $S_i$ on some vertex $v'$. Since $e$ was not removed in Step~3, $e' \notin M_i$. However, since $v' \in e'$, $e'$ is removed in Step~3. Thus, every vertex $v$ in $G_i$ has at least one incident edge removed from $G_{i-1}$. Combining these observations, we find:

    \begin{itemize}
        \item every edge $e \in G_i$ is disjoint from every edge in $M_1 \cup M_2 \cup \cdots \cup M_i$, and
        \item the maximum degree of $G_i$ is $\Delta - i$.
    \end{itemize}

    By the second point above, $G_\Delta$ is an independent set of vertices. Also, by the first point and the observation that each $M_i$ is a matching, we find that $M = \bigcup_{i=1}^\Delta M_i$ is a matching as well. Finally, $M$ is maximal because every removed edge intersects some edge in $M$.

    The procedure invokes $A$ to find a $1$-weak MIS $\Delta$ times, and after each invocation, only $O(1)$ rounds of communication are performed. Therefore, the total running time of the procedure is $O(\Delta T)$. By Theorem~\ref{thm:mm_lb}, the running time must be $\Omega(r \Delta)$, hence $T = \Omega(r)$. 
\end{proof}

The lower bound of $\Omega(r)$ is perhaps surprising given the upper bounds for $r-1$-weak MIS in~\cite{Balliu2023-distributed} and our Theorem~\ref{thm:k-weak-mis} whose $r$-dependence is sublinear.

% \paragraph*{$\Omega(\Delta + \log^* n)$ Lower Bounds}
\paragraph*{Lower Bounds in $\Delta$}

We compliment our $\Omega(r + \log^* n)$ lower bound with a pair of $\Omega(\Delta + \log^* n)$ lower bounds on finding a MIS in Hypergraphs, complementing the above bounds. First, we show an $\Omega(\Delta + \log^* n)$ lower bound for finding a $1$-weak independent set on the family of hypergraphs of even rank. Secondly, we present an $\Omega(\Delta + \log^* n)$ lower bound for finding a $k$-weak independent set for odd values of $k$ on the family of hypergraphs of rank $2k$. Both are due to a reduction from the $\Omega(\Delta + \log^* n)$ lower bound on finding an MIS on regular graphs due to Balliu et al. \cite{balliu2021lower}.

\begin{theorem}
    \label{thm:delta-lb}
    Consider the family of $r$-uniform hypergraphs with maximum degree $\Delta$ such that $r$ is even. Then in this family, computing a $1$-weak MIS requires $\Omega(\Delta + \log^* n)$ rounds in the local model.
\end{theorem}

\begin{proof}
    Consider a (non-hyper) graph $G = (V, E)$ of maximum degree $\Delta$. We construct a hypergraph $G' = (V', E')$ as follows. For each vertex $v \in V$, we create a set of $r / 2$ vertices $v_1, v_2, \dots, v_{r/2}$ in $V'$. For each edge $(v, u) \in E$, we add the edge $(v_1, v_2, \dots, v_{r / 2}, u_1, u_2, \dots$, $u_{r/2})$ to $E'$. Observe that a $1$-weak MIS in $G'$ can be converted into an MIS for $G$ as follows. Consider some $v_i \in V'$ such that $v_i$ is in the MIS, then, for any $u_j \in V'$ such that there exists some edge $e \in E'$ for which $v_i, u_j \in e$, $u_j$ is not in the MIS. Therefore, by adding to the MIS for $G$ the vertex $v$, for any $v$ such that some vertex $v_i$ is in the MIS of $G'$, we get a maximal independent set of $G$. As there is known to be a lower bound of $\Omega(\Delta + \log^* n)$ for finding such a lower bound for graphs, we get the stated lower bound.
\end{proof}

We can use the same technique to obtain a $\Omega(\Delta)$ lower bound for $k$-weak MIS for an odd value of $k$, and hypergraph of rank $(2k)$. The high-level idea is the same as above, with the difference that instead, we construct the hypergraph with $k$ copies of each vertex, and use the odd parity of $k$ to only add vertices in the original graph $G$ to the independent set of $G$ if at least $k/2$ of the corresponding vertices in the hypergraph belong to the corresponding $k$-weak independent set. This allows a more general bound, relating the rank and weakness of the set in a closer manner than before. We now formalize our construction.

\noindent
\textbf{Construction.} Given a graph $G = (V, E)$, and $k \in \mathbb{N}$ where $k$ is odd, we construct the hypergraph $H = (V', E')$ as follows. For each vertex $v_i \in V$, we add $k$ vertices into $V'$, labelled $v_{i, 1}, v_{i, 2}, \dots, v_{i, k}$. For each edge $(v_i, v_j) \in E$, we add the hyperedge $(v_{i, 1}, v_{i, 2}, \dots, v_{i, k}, v_{j, 1}, v_{j, 2}, \dots, v_{j, k})$ to $E'$.

\begin{theorem}
    \label{thm:delta-lb-k}
    Consider the family of $2k$-uniform hypergraphs with maximum degree $\Delta$ such that $k$ is odd. Then in this family, computing a $k$-weak MIS requires $\Omega(\Delta + \log^* n)$ rounds in the local model.
\end{theorem}

\begin{proof}
    Let $H$ be constructed as above for the graph $G$, and let $S$ be a $k$-weak maximal independent set in $H$. Observe that each vertex $v_{i, \ell}$ must belong to at least one edge with $k$ vertices in the independent set. Further, as $\{ v_{i, \ell} \mid \ell \in [1, k]\} \subset e$ or $\{ v_{i, \ell} \mid \ell \in [1, k]\} \cap e = \emptyset$, for every $e \in E'$, we have that all vertices in $\{ v_{i, \ell} \mid \ell \in [1, k]\}$ share the same set of edges containing $k$ vertices in $S$.

    Now, let $e$ be some edge such that $\vert e \cap S \vert = k$, and let $e = (v_{i, 1}, v_{i, 2}, \dots, v_{i,k}, v_{j, 1}, v_{j, 2}$, $\dots, v_{j, k})$. Observe that, as $k$ is odd, either $\vert \{v_{i, 1}, v_{i, 2}, \dots, v_{i,k}\} \cap S \vert > k/2$ or $\vert \{v_{j, 1}, v_{j, 2}, \dots$, $v_{j,k}\} \cap S \vert > k/2$. % Let us assume that $\vert \{v_{i, 1}, v_{i, 2}, \dots, v_{i,k}\} \cap S \vert > k/2$.

    Consider the set $S' \subseteq V$ in the original graph $G$, where $v_i \in S$ if and only if $\vert \{v_{i, 1}, v_{i, 2}, \dots$, $v_{i,k}\} \cap S \vert > k/2$. Observe that, given any pair $v_i, v_j$ where $(v_i, v_j) \in E$, and $v_i \in S'$, we have that $\vert \{v_{j, 1}, v_{j, 2}, \dots, v_{j,k}\} \cap S \vert < k/2$, by definition of a $k$-weak independent set in $H$. Therefore, by construction, $v_j \notin S'$. In the other direction, if $v_j \notin S'$, then observe that there must be some edge $e \in E'$ where $\vert e \cap S \vert = k$ and $\{v_{j, 1}, v_{j, 2}, \dots, v_{j,k}\} \subset e$. Now, let $e = (v_{j, 1}, v_{j, 2}, \dots, v_{j,k}, v_{\ell, 1}, v_{\ell, 2}, \dots, v_{\ell,k})$ be such an edge, and observe that $\vert \{v_{\ell, 1}, v_{\ell, 2}, \dots, v_{\ell,k} \} \cap S \vert > k/2$. Therefore, $(v_j, v_{\ell}) \in E$ and $v_{\ell} \in S'$. Thus, for any $v_j$ not in $S'$, there exists at least one vertex adjacent to $v_j$ in $G$ that is in $S'$ and, conversely, for any $v_i \in S'$, no vertex adjacent to $v_i$ belongs to $S'$. Hence $S'$ is an independent set, and thus, by the previous bound of $\Omega(\Delta + \log^* n)$ rounds on finding an MIS in $G$ holds for finding a $k$-weak independent set on $H$.
\end{proof}

\section{LLL formulation}
\label{sec:lll}

With our lower bounds in mind, we now begin to provide positive results to the problem.
We first present a formulation of this problem using the Lovasz Local Lemma as a means to find an $(\alpha, \beta)$-independent set for some range of parameters.

Recall that the Lovasz Local Lemma (LLL)~\cite{erdos1975problems} gives sufficient conditions for a probability space to have an outcome that avoids a set of ``bad'' events. More formally, suppose $B_1, B_2, \ldots, B_m$ are events in a probability space $\Omega$ such that each $B_i$ is independent of all but $d$ other $B_j$ for $j \neq i$. The LLL asserts that if for all $\Pr(E_i) \leq p$ and $e p (d+1) \leq 1$, then there exists an outcome $\omega \in \Omega$ that is not contained in any of the $E_i$. While the original formulation is non-constructive, an efficient constructive, and distributed solution was found by Moser and Tard\"os in \cite{Moser2010-constructive}. In the distributed setting, each event $B_i$ is represented by a vertex in the network and $B_i$ and $B_j$ share an edge if and only if the events $B_i$ and $B_j$ are not independent. An \emph{LLL formulation} of a distributed problem $P$ is a reduction from $P$ to the LLL.\looseness=-1

The idea of our LLL formulation is that we associate with each vertex $v \in V$ with a $\set{0, 1}$-random variable with the interpretation that $v = 1$ if and only if $v$ is contained in the independent set. Thus, an outcome (i.e., an assignment of $\set{0,1}$-values to each $v$) corresponds to an $(\alpha, \beta)$-MIS if and only if (1)~for each edge $e \in E$, $\sum_{v \in e} v \leq \beta$, and (2)~for each vertex $v \in V$, either $v = 1$ or $v$ is contained in an edge $e$ with $\sum_{u \in e} u \geq \alpha$. In the LLL formulation, events $B$ correspond to violations of (1) and (2). We make the following observations assuming that $H = (V, E)$ has rank $r$ with maximum degree at most $\Delta$:
\begin{enumerate}
    \item Type~1 events $B_e$ and $B_f$ are independent unless their corresponding edges $e$ and $f$ intersect. Therefore $B_e$ is independent of all but $\Delta r$ type~1 events.
    \item Type~2 events $B_v$ and $B_u$ are independent unless their corresponding vertices $v$ and $u$ are neighbors. Thus, again, $B_v$ is independent of all but at most $\Delta r$ other such type~2 events.
    \item A type~1 event $B_e$ and type~2 event $B_v$ are independent unless $v \in e$. Thus, $B_e$ is independent of all but $r$ type~2 events, and $B_v$ is independent of all but $\Delta$ type~1 events.
    \item Every event is a function of the $1$-hop neighborhood of some vertex $v \in V$. Therefore, communication the LLL network can be simulated in the underlying graph with one additional communication round in the LOCAL model.
\end{enumerate}
Combining the observations 1--3 above, we obtain the following.

\begin{lemma}\label{lem:lll-degree}
    Consider the LLL formulation described above where $H$ is a hypergraph with rank $r$ and maximum degree $\Delta$. Then each event $B$ is independent of all but $d \leq \Delta r + \Delta + r \leq 2 \Delta r - 1$ other events.
\end{lemma}

In order to complete the LLL formulation of finding an $(\alpha, \beta)$-MIS, we must define a probability measure on the outcomes such that probabilities $p$ of the events satisfy $e p (d+1) \leq 1$. To this end, we prove the following lemma.

\begin{lemma}\label{lem:lll-formulation}
    Suppose $H$ is a hypergraph with rank $r$ and maximum degree $\Delta$. Suppose $\alpha, \beta$ with $0 < \alpha \leq \beta \leq r - 1$ satisfy
    \begin{equation}\label{eqn:lll-bound}
        \frac{(\beta - \alpha)^2}{\beta + \alpha} \geq 6 \log(16 r \Delta).
    \end{equation}
    Then there is an LLL formulation of finding $(\alpha, \beta)$-independent sets in $H$.
\end{lemma}

\begin{proof}
    Let $\mu = \frac{\beta + \alpha}{2}$ denote the midpoint between $\alpha$ and $\beta$. Consider the probability measure determined by choosing each $v = 1$ with probability $q = \mu / r$ (and $v = 0$ with probability $1-q$) independently. For each edge $e$, let $X_e = \sum_{v \in e} v$, so that $\E(X_e) = \mu$. We make the following observations:
    \begin{itemize}
        \item if $X_e \leq \beta$, then the event $B_e$ does not occur
        \item if $X_e \geq \alpha$ the for every $v \in e$, the event $B_v$ does not occur
    \end{itemize}
    Observe that both conditions above are satisfied if $\abs{X_e - \mu} \leq \frac{\beta - \alpha}{2}$. To bound the probability that the latter condition does not hold, we apply the following Chernoff bound \cite{Mitzenmacher_Upfal_2005}: if $X$ is a sum of independent $0$-$1$ random variables with expected value $\mu$, then for $\delta$ satisfying $0 < \delta < 1$, we have
    \begin{equation}\label{eqn:chernoff}
        \Pr(\abs{X - \mu} > \delta \mu) \leq 2 \exp(-\mu \delta^2 / 3).
    \end{equation}
    Applying~(\ref{eqn:chernoff}) with $\delta = \frac{\beta - \alpha}{\beta + \alpha}$, we obtain
    \begin{equation}\label{eqn:bad-bound}
        \Pr\paren{\abs{X - \mu} > \frac{\beta - \alpha}{2}} \leq 2 \exp\paren{-\frac{(\beta - \alpha)^2}{6(\beta + \alpha)}}.
    \end{equation}
    This expression is an upper bound on the probability, $p$ of any bad event $B$. Therefore we find that
    \begin{align*}
        e p (d+1) &< 4 \cdot 2 \exp\paren{-\frac{(\beta - \alpha)^2}{6(\beta + \alpha)}} \cdot (2 \Delta r)
        \leq 16 \Delta r \exp(-\ln(16 \Delta r)) = 1,
    \end{align*}
    where we applied~(\ref{eqn:lll-bound}) in the second inequality.
\end{proof}

We may then apply a distributed formulation of Moser and Tard\"os's constructive LLL algorithm~\cite{Moser2010-constructive} to obtain the following:

\begin{corollary}\label{cor:lll}
    There exists a deterministic distributed algorithm that finds an $(\alpha, \beta)$-MIS in an hypergraph $H$ with rank $r$ and maximum degree $\Delta$ in $O(\log^2 n)$ rounds of the LOCAL model for any $\alpha, \beta$ satisfying $\frac{(\beta - \alpha)^2}{\beta + \alpha} \geq 6 \log(16 r \Delta)$. In particular, these conditions are satisfied for:
    \begin{itemize}
        \item $\alpha = 1$ and $\beta \geq 18 \log(16 r \Delta)$, and
        \item $\alpha = cr$ and $\beta \geq \alpha +  c' \sqrt{r \log(r \Delta)}$.
    \end{itemize}
\end{corollary}

\subsection{A Zero Round Protocol}

We note that for larger values of $\beta - \alpha$, the concentration inequality in~(\ref{eqn:lll-bound}) is strong enough to guarantee that the set $S$ formed by selecting each vertex independently to be in $S$ with probability $q = \mu/r$ is an $(\alpha, \beta)$-MIS with high probability. In particular, we get the following:

\begin{corollary}\label{cor:zero-round}
    Suppose $\alpha$ and $\beta$ satisfy 
    \begin{equation}\label{eqn:zero-round}
        \frac{(\beta - \alpha)^2}{\beta + \alpha} \geq 6 c \log n + 6.
    \end{equation}
    Then there exists a zero-round randomized distributed algorithm that finds an $(\alpha, \beta)$-MIS with probability at least $1 - \frac{1}{n^c}$. In particular, this condition is satisfied for:
    \begin{itemize}
        \item $\alpha = 1$ and $\beta \geq 18 c \log n + 18$
        \item $\alpha = c' r$ and $\beta \geq \alpha + c'' \sqrt{r \log n}$
    \end{itemize}
\end{corollary}

\section{Finding Independent Sets in Hypergraphs of large rank}
\label{sec:high-rank}

We now present our result using a randomized algorithm to find weak independent sets of expected size at least $r/2k$ for $r \geq \sqrt{n}$ and $\Delta$ being $o(r)$. This result applies to both linear hypergraphs and to \emph{$\lambda$-intersecting} hypergraphs, hypergraphs where the size of the intersection between any pair of edges is no more than some constant value $\lambda$.

\begin{definition}[Linear hypergraphs]
    \label{def:linear_hypergraph}
    A hypergraph $G = (v, E)$ is \emph{linear} if, for every $e, f \in E$ where $e \neq f$, $\vert e \cap f \vert \leq 1$.
\end{definition}

\begin{definition}[$\lambda$-intersecting hypergraphs]
    A hypergraph $G = (V, E)$ is $\lambda$-intersecting if, $\forall e_1, e_2, \dots, e_{\lambda + 1} \in E$ where $e_i \neq e_j, \forall 1 \leq i < j \leq \lambda + 1$, it holds that $\vert \bigcap_{\ell \in [1, \lambda + 1]} e_{\ell} \vert \leq 1$.
\end{definition}

Note that a linear hypergraph is a $1$-intersecting hypergraph, and any $1$-intersecting hypergraph is a linear hypergraph.

\paragraph*{Algorithm outline.}
Let $G = (V, E)$ be a $\lambda$-intersecting hypergraph, and $S$ be a set of vertices, initially containing all vertices in $V$. Our algorithm operates by iteratively removing vertices from $S$ until $S$ forms an independent set.

% We start with a rank $r$ $\lambda$-linear hypergraph $G = (V, E)$ and set $S$. Initially, every vertex in $V$ belongs to the set $S$, with vertices removed from $S$ as the algorithm progresses until $S$ forms a weak independent set.

% We partition each edge $e$ into $r / k$ disjoint sets of size $k$, labeled $u_{e, 1}, u_{e, 2}, \dots, u_{e, r/k}$. Observe that every set constructed in this way satisfies $\vert \{ f \in E \mid u_{e, i} \subseteq \}$

We partition the edge $e$ into $r/k$ disjoint sets of size $k$, which we label $u_{e, 1}, u_{e, 2}, \dots, u_{e, r/k}$. Observe that, for any $k > 1$, there can be at most $\lambda$ edges containing $u_{e, i}$ as a strict subset.
% Observe that every set $u_{e, i}$ constructed in this way satisfies $\vert \{f \in E \mid u_{e, i} \subseteq f \} \vert \leq \lambda$, i.e. there exists at most $\lambda$ edges containing $u_{e, r/k}$ as a subset.

Now, each edge $e \in E$ selects randomly some vertex $v \in e$ to remove from $S$. 
If $\lambda < k$, for every $u_{e, i}$, $\vert u_{e, i} \cap S \vert \geq 1$, i.e. there is at least one vertex left in $u_{e, i}$ that is in the candidate independent set.

The vertices of each edge that are not removed stay in the independent set.

\paragraph*{Analysis}

% \begin{lemma}
% In a $\lambda$-linear hypergraph it holds that $\lambda \leq \Delta k$
    
% \end{lemma}

% \begin{proof}
% Each $v$ in $u_{e,i}$ is contained in at most $\Delta$ hyperedges , so no more that $\Delta k$ hyperedges can contain $u_{e,i}$ and this holds for every $i$ from $1$ to $r/k$ .
    
% \end{proof}

Consider some vertex $v_0$ in the set $u_{e,i}$.
$v_0$ escapes $e'$ where $e'$ is a neighbor of $e$ 
(meaning that $e'$ does not select $v_0$ to remove it) iff (a) $e'$ does not select any node in $u_{e,i}$ to remove or (b) given $e'$ selects some node in $u_{e,i}$ to remove but it does not select $v_0$.

The probability of (a) is $1-r/k$ and the probability of (b) is $k/r (1-1/k)$. Therefore the probability that $v_0$ escapes $e'$ is $1-k/r + k/r(1-1/k)$ which is exactly $1-1/r$. Therefore, the probability that $v_0$ escapes all $e'$ that contain it is $(1-1/r)^{\lambda}$. Noting that $\lambda \leq k \Delta$, when $\Delta$ is $o(r)$ we have that $(1 - 1/r)^{\lambda} \geq (1 - 1/r)^{k \Delta}$. By extension, the expected number of vertices in $e$ that escape is , by linearity of expectation , at least $r/k (1-(k \Delta)/r$ i.e. at least $r/(2k)$.

% So, the probability that $v_0$ escapes all $e'$ that contain it is $(1-1/r)^\lambda$ . This is at least $(1-1/r)^{\Delta k}$ . This is at least $1-(\Delta k)/r$ when $r \geq \sqrt{n}$ and $\Delta$ is $o(r)$ .

% Therefore the expected number of vertices in $e$ that escape is , by linearity of expectation , at least $r/k (1-(\Delta k)/r$ i.e. at least $r/(2k)$.

\begin{lemma}
\label{lem:high-rank}
The simple randomized algorithm presented above produces an $m$-weak independent set with expected size of $m$ at least $r/2k$.
\end{lemma}

%Since the events that a particular vertex does not escape are negatively dependent, the above size of the weak independent set holds also with high probability by concentration bounds around the expected value.

\section{Deterministic Algorithms}
\label{sec:deterministic-algorithms}

We now provide a set of deterministic algorithms to find a $(\alpha, \beta)$-IS, a $(2, k)$-Ruling set, or a $k$-MIS. All three approaches are based on the idea of finding \emph{defective colorings} of the graph. A $k$-defective coloring is a coloring such that no more than $k$ vertices in a given edge share the same color.

\subsection{Finding an $(\alpha, \beta)$-IS in $O(\Delta r / (\beta-\alpha))$ rounds}
\label{sec:alpha-beta-mis-alg}

% [{\color{red} TODO (Duncan): I think I got $\beta$ and $\alpha$ a bit muddled in the definition. I will come back and fix shortly}]
% [{\color{blue} (Will): It looks like this is exactly the same algorithm that I described in our meeting, so nothing substantive to change there. I think it is more general to define the running time as a function of $\delta = \beta - \alpha$. In this case the running time is $\Delta r / \delta$. I would also suggest framing the algorithm as follows: Define a $\delta$-edge defective coloring to be a vertex coloring such that the intersection of each color class with each edge has cardinality at most $\delta$. Given any such coloring with $k$ colors, we can compute an $(\alpha, \beta)$-MIS in $k$ rounds for any $\delta \leq \beta - \alpha$. Then observe that we can obtain a $\delta$-edge defective coloring in $O(r \Delta / \delta + \log^* n)$  rounds using the edge partitioning trick together with known coloring algorithms. I think breaking the process up into two steps makes it more clear that any improvement on the coloring front gives an immediate improvement for $(\alpha, \beta)$-MIS computation.
% }]

We next provide a deterministic algorithm for finding an $(\alpha, \beta)$-IS in $O(\Delta r / (\beta - \alpha + 1))$ rounds. For the remainder of this section we assume that we have a hypergraph $H = (V, E)$ in the LOCAL model, of rank $r$. Further, we assume that every vertex $v$ knows the values of $\beta$ and $\alpha$ used, and the edges incident to $v$. We define the parameter $\delta = \beta - \alpha + 1$.
% Throughout this section, we use the variable $\delta = \beta - \alpha$ for notational convenience. Therefore, our algorithm requires $O(\Delta r / \delta)$ rounds. Importantly, if $\beta = c r$ for some constant $c$, and thus $\delta \geq c r / 2$, then this requires only $O(\Delta)$-rounds.
% Note that, in general,  $O(\Delta (r / \alpha))$ is significantly more efficient than the trivial $O(\Delta r)$ round algorithm, particularly as the value of $\alpha$ increases.

\noindent
\textbf{Outline.} Our algorithm operates in a two-stage manner. First, we find a \emph{$\delta$-edge defective coloring}, a coloring where no edge contains more than $\delta$ colors of a given class. Then, we iterate through each class, activating any vertices of the given color, which can then choose to add themselves to the independent set or not.

To determine the $\delta$-defective edge coloring, we have each edge partition itself into $\delta$ sets of size $r / \delta$. We construct a new graph by replacing each hyperedge with the new sets, thus getting a hypergraph of rank $r/\delta$, and finding an $\psi = O(\Delta r / \delta)$ coloring on the underlying graph of this new hypergraph. Once this coloring is found, we iterate through each color class from $1$ to $\psi$, activating at round $i$ all vertices with the $i^{th}$ color. Any active vertex $v$ that belongs only to edges with fewer than $\alpha$ vertices in the independent set will add itself to the set. Otherwise, $v$ excludes itself.

The key observation is that after the edge coloring, for any edge containing the vertex $v$, there are at most $\delta$ vertices sharing the same color as $v$. 
Therefore, if every edge incident to $v$ contains less than $\alpha$ vertices in the independent set, $v$ can add itself to the set, without any risk of the number of any edge incident to $v$ containing more than $\beta$ vertices in the independent set. We note that in the case of $\alpha = \beta$ (hence $\delta = 1$), this algorithm is equivalent to the ``trivial'' algorithm for finding a $\beta$-weak MIS suggested in Lemma~\ref{lem:delta-r-algorithm}. 

We now formalize our approach.

\noindent
\textbf{Finding a $\alpha$-defective coloring.}
For each edge $e \in E$, we construct $\delta$ sets of size $r/\delta$, $e_1, e_2, \dots, e_{\delta}$. For simplicity, if $E = \{v_1, v_2, \dots, v_{r}\}$ where $ID(v_1) < ID(v_2) < \dots < ID(v_{\delta})$, we have $e_i = \{v_{(i - 1) r/\delta + 1}, v_{(i - 1) r/\delta + 2}, \dots, v_{i r/\delta}\}$. Let $H' = \{V, E'\}$ be the hypergraph formed by replacing each edge in $e \in E$ with the sets $e_1, e_2, \dots, e_{\delta}$, formally $E' = \bigcup_{e \in E} \{e_1, e_2, \dots, e_{\delta}\}$. Correspondingly, let $G'$ be the underlying graph of $H'$.
Now, let $\mathbf{\Psi} : V \mapsto 1, 2, \dots, \Delta r / \delta + 1$ be a coloring of the underlying graph $G'$, corresponding to a maximum degree + 1 coloring of $G'$. Note that as $G'$ is a regular graph, $\mathbf{\Psi}$ can be found in $O(\sqrt{(\Delta r / \delta) (\log \Delta r / \delta}) + \log^* n)$ rounds by the algorithm due to Maus and Tonoyan~\cite{maus2020local}. Now, observe that as each edge in $H'$ is properly colored, and the edge $e \in E$ in the original hypergraph contains $\alpha$ such sets, for any color $\chi \in [\Delta r / \delta + 1]$, $\vert e \cap \{ v \in V \mid \mathbf{\Psi}(v) \}\vert \leq \delta$. Thus $\mathbf{\Psi}$ is a $\delta$-defective coloring of $H$.

\begin{lemma}
    \label{lem:defective_coloring}
    Given a hypergraph $H$ of rank $r$, and value $\delta \leq r$, Algorithm \ref{alg:defective_coloring} finds a $\delta$-defective coloring of $H$ using $\Delta r/\delta + 1$ colors can be found in $O(\sqrt{(\Delta r / \delta) (\log \Delta r / \delta}) + \log^* n)$ rounds.
\end{lemma}

\begin{proof}
    Follows from above.
\end{proof}

\begin{algorithm}
    \caption{$\delta$-defective coloring algorithm on the graph $H = (V, E)$ of rank $r$ for vertex $v \in V$. We assume $E$ only contains the edges incident to $v$.}
    \label{alg:defective_coloring}
    \begin{algorithmic}[1]
        \Procedure{$\delta$-Defectivecolor}{}
            \State Split each edge $e$ into $c$ subsets $e_1, e_2, \dots, e_{\delta}$, with $v$ in $e_i$.
            \State $H' \gets (V, \bigcup_{e \in E} \{e_1, e_2, \dots, e_{\delta}\})$
            \State $G' \gets$ Underlying Graph of $H'$
            \State $\mathbf{\Psi} \gets (\Delta r / \delta)$-\textsc{color($G', v$)}
            \State \textbf{Return} $\mathbf{\Psi}$.
        \EndProcedure
    \end{algorithmic}
\end{algorithm}

\noindent
\textbf{Forming an independent set.}
Using the $\delta$-defective coloring $\mathbf{\Psi}$, we now find the independent set. We iterate through the set of colors, activating at each round the set of vertices in the given color class. When active, each vertex is given the chance to add itself to the independent set, or to mark itself as inactive. Formally, at round $i$, any vertex $v$ where $\mathbf{\Psi}(v) = i$ will add itself to the independent set $S$ if and only if $\forall e \in E$ either $v \notin e$ or $v \in e$ and $\vert e \cap S \vert < \alpha$. Therefore, after $v$ has been activated, either $v$ is in the independent set $S$, or $v$ is incident to at least one edge containing $\alpha$ vertices in the independent set, locally satisfying the properties of an $(\alpha, \beta)$-IS.

\begin{algorithm}
    \caption{$O(\Delta r / \delta)$ round coloring algorithm for some vertex $v$ in the hypergraph $H = (V, E)$.}
    \label{alg:EdgePartitionIS}
    \begin{algorithmic}[1]
        \Procedure{EdgePartitionIS}{Hypergraph $H$, $(\alpha, \beta) \in \mathbb{N}^2$}
            \State $\mathbf{\Psi} \gets \delta$-\textsc{Defectivecolor}$()$
            \State $\psi \gets \vert \mathbf{\Psi} \vert$
            \State $A \gets \emptyset$ \% \emph{Set of vertices adjacent to $v$ in the independent set.}
            \For{$i \in [\psi]$}
                \If{$\mathbf{\Psi}(v) = i$}
                    \If{$\forall e \in \{v \in e \mid e \in E\}$, $\vert e \cap A \vert \leq \alpha$}
                        \State Mark as in the independent set
                        \State Broadcast that $v$ is adding itself to the independent set.
                    \Else
                        \State Mark as not in the independent set
                        \State Broadcast that $v$ is not adding itself to the independent set.
                    \EndIf
                \Else
                    \State Add to $A$ any adjacent vertex that adds itself to the independent set.
                \EndIf
            \EndFor
        \EndProcedure
    \end{algorithmic}
\end{algorithm}

\begin{theorem}
    \label{thm:edgepartitionIS}
    Algorithm \ref{alg:EdgePartitionIS} constructs an $(\alpha, \beta)$-IS in $O(\Delta r / \delta + \log^*n)$ rounds in any hypergraph $H = (V, E)$ of rank $r$.
\end{theorem}

\begin{proof}
    We prove this by an inductive argument. First, not that any vertex $v$ where $\mathbf{\Psi}(v) = 1$ will, by construction, add itself to the independent set $S$. As there can be at most $\alpha$ vertices in any edge colored $1$, no edge can, at the end of this step, violate the condition that no edge contains more than $\beta$ vertices in the independent set.

    Now consider some vertex $v$ where $\mathbf{\Psi}(v) = i$, and assume that after the first $i - 1$ color classes have been activated, no edge contains more than $\beta$ vertices in the independent set. If $v$ is incident to at least one edge with more than $\alpha$ vertices in the independent set, then $v$ may not add itself to the independent set. Note that $v$ is locally satisfied as it belongs to at least one edge containing at least $\alpha$-vertices in the independent set. Otherwise, if no such edge is incident to $v$, $v$ adds itself to the independent set. As no edge contains more than $\delta = \beta - \alpha + 1$ vertices in the color class $i$, and a vertex will only add itself to the independent set if there are fewer than $\alpha$ vertices in the independent set in every edge it is incident to, no edge can, after color class $i$ has been activated, contain more than $\beta$ vertices in the independent set. Thus after being activated, $v$ is either in the independent set, or incident to one edge with at least $\alpha$ vertices in the independent set. In either case, the locally checkable properties of the set are satisfied. By inductive argument, we have the claim of correctness.

    To determine the round complexity, observe that the $\alpha$-defective coloring can be found in $O(\sqrt{(\Delta r / \alpha) (\log \Delta r / \alpha}) + \log^* n)$ rounds. As this coloring uses at most $\Delta r / \alpha + 1$ colors, the second step of the algorithm requires $O(\Delta r / \alpha)$ rounds, giving the total complexity of $O(\Delta r / \alpha + \log^*n)$.
\end{proof}

\paragraph*{$(2, k)$-Ruling Sets.}
\label{sec:ruling-sets}

We note that our $(\alpha, \beta)$-IS algorithm may be adopted to find a $(2, k)$-Ruling sets in $O(\Delta (k + (2^{k - 1} / 3^{k - 1})r) + \log * n)$ rounds via a recursive approach for the hypergraph $H = (V, E)$. We begin with the special case where $k = O(\log r)$, requiring $O(\Delta \log r)$ rounds, noting that the arguments are similar in both cases.

We find a $(2, O(\log(r)))$-ruling set by first computing a $(2 r/3, r / 3)$-IS using Algorithm \ref{alg:EdgePartitionIS}, $S_1$. Using $S_1$, we take the graph $H_1 = (S_1, E_1)$, where $E_1 = \{e \cap S \mid e \in E \}$, of maximum rank $r_1 = 2 r / 3$, and find a $(4 r / 9, r / 9)$-IS (equivalently, a $(2 r_1 / 3, r_1 / 3)$-IS), $S_2$. We repeat this process $O(\log r)$ times, until the rank of the graph induced by the independent set is some constant $c$, at which point we find an $O(\Delta)$-coloring of the underlying graph in $O(\sqrt{\Delta \log \Delta} + \log^* n)$ rounds by the algorithm due to Maus and Tonoyan~\cite{maus2020local}, then iterate through each color class in order, with the vertices adding themselves to the set if no neighboring vertex belongs to it.

\begin{theorem}
    Algorithm \ref{alg:FindRS} finds a $(2, O(\log r))$-Ruling Set in a hypergraph $H = (V, E)$ of maximum degree $\Delta$ and rank $r$ in $O(\Delta \log r + \log * n)$ rounds in the LOCAL model of computation.
\end{theorem}

\begin{proof}
    Let $H_1 = (V, E_1)$, and let $S_1$ be the $(2 r / 3, r / 3)$-independent set found by Algorithm \ref{alg:EdgePartitionIS} on $H_1$. Note that, by Theorem \ref{thm:edgepartitionIS}, this requires $O(\Delta)$ rounds. In general, let $H_i = (S_{i - 1}, E_{i})$, where $E_{i} = \{e \cap S_{i - 1} \mid e \in E\}$, and $S_{i - 1}$ is the independent set found by Algorithm \ref{alg:EdgePartitionIS} on $H_{i - 1}$.
    
    We claim that every vertex $v \in V$ is a distance of at most $i$ from at least one vertex $u \in S_i$. First, consider the $S_1$. Observe that, by Algorithm \ref{alg:EdgePartitionIS}, every vertex must either be in the independent set, or adjacent to at least $r / 3$ vertices in this set. Thus, the statement holds. Assume now that this holds for every $j \in [i - 1]$. Then, observe that for any vertex $v \in S_{i - 1} \setminus S_i$ must be adjacent to some vertex $v'$ in $S_i$. Hence, any vertex $u$ at a distance of $i - 1$ from $v$ will be at a distance of at most $i$ from $v'$, and hence the claim is satisfied.

    Now, observe that after $O(\log r)$ rounds, we will have an $(1, 1)$-IS $S$ in $H$. Therefore, as no pair of vertices $S$ can be adjacent, $S$ must correspond to a $(2, O(\log r))$-ruling set.
\end{proof}

\begin{corollary}
    Algorithm \ref{alg:FindRS} finds a $(2, k)$-Ruling Set in a hypergraph $H = (V, E)$ of maximum degree $\Delta$ and rank $r$ in $O(\Delta (k + (2^{k - 1} / 3^{k - 1})r) + \log * n)$ rounds in the LOCAL model of computation.
\end{corollary}

\begin{proof}
    Using the same approach as above, we do $k - 1$ rounds of the recursive process, ending with the $((2^k / 3^k)r, r / 3^k)$-IS $S$ such that every vertex $v \in V \setminus S$ is at a distance of at most $k - 1$ from at least one vertex in $S$. Now, we compute an MIS $I$ in the graph $G$ corresponding to the underlying graph of $(S, \{e \cap S \mid e \in E \})$, requiring $O(\Delta (2^{k - 1} / 3^{k - 1})r)$ time. Note that, by the same arguments as above, every vertex $v \in V \setminus I$ is a distance of at most $k$ from at least one vertex $u \in I$. Thus, $I$ forms a $(2, k)$-Ruling Set of $H$.
\end{proof}

\begin{algorithm}
    \caption{Algorithm for finding a $(2, k)$-Ruling Set in the hypergraph $H$.}
    \label{alg:FindRS}
    \begin{algorithmic}
        \Procedure{FindRS}{Hypergraph $H$, $k \in \mathbb{N}$}
            \State $InIS \gets True$ \% \emph{Initially, we say every vertex is in the independent set.}
            \State $I \gets V$ \% \emph{Add every vertex to the independent set.}
            \State $r' = r$ \% \emph{Current}
            \While{$InIS$ and $r' > 2$}
                \State $E' \gets \{e \cap I \mid e \in E \}$ \% \emph{Update the edges with the neighbors of $v$ still in the independent set.}
                \State \textsc{EdgePartitionIS}($(I, E')$, $(2 r' / 3, r' / 3)$)
                \State $r' \gets 2 r' / 3$
                \If{$v$ not in the independent set}
                    \State $InIS \gets False$
                \EndIf
                \State Update independent set $I$.
            \EndWhile
            \If{$InIS$}
                \State $E' \gets \{e \cap I \mid e \in E \}$ \% \emph{Update the edges with the neighbors of $v$ still in the independent set.}
                \State \textsc{FindIS}$($Underlying Graph $(I, E')$ $)$ \% \emph{Using known algorithms, this wil take $O(r' \Delta + \log^* n) = O(\Delta + \log^* n)$ rounds.}
            \EndIf
        \EndProcedure
    \end{algorithmic}
\end{algorithm}

\subsection{Finding $k$-MIS for large $k$}
\label{sec:k-mis}

Here, we describe an algorithm that finds a $k$-weak MIS in in $O(\Delta^2 (r - k) \log r + \Delta \log r \log^* r + \log^* n)$ rounds. Note that in the regime where $r$ is much larger than $\Delta$ and $k$ is close to $r$, this gives an asymptotic improvement over the ``trivial'' $O(r \Delta + \log^* n)$ algorithm described in Section~\ref{sec:trivial-algorithms}. This result generalizes of a result of Balliu et al.~\cite{Balliu2023-distributed}, which gives an $O(\Delta^2 \log r + \log r \log^*r + \log^* n)$ time algorithm for finding an $(r-1)$-weak MIS, which is comparable our result for $k = r - 1$. 

The idea of our algorithm is to iterate the procedure described of Section~\ref{sec:alpha-beta-mis-alg} to find an $(\alpha, \beta)$-MIS $\Delta \log r$ times with $\beta = k$. On the first iteration, we obtain at $(k/2, k)$-MIS. In subsequent iterations, we consider only ``active'' vertices that have not yet joined the IS. In each iteration of the procedure, each edge $e$ with $a_e$ active vertices (and $i_e = r - a_e$ elements in the IS) partitions its active vertices, and the active vertices color themselves to give an edge-defective coloring as in Section~\ref{sec:alpha-beta-mis-alg}. This is done such that each color class contains at most $(k - i_e)/2$ colors, hence adding all vertices of any one color class will not violate the $k$-independence of the IS. Given this coloring, we then iterate over color classes, and each vertex adds itself to the IS greedily as in the previous section.

After each iteration of the procedure above, we observe that the following holds: for each vertex $v$, either $v$ is added to the IS, or $v$ is incident to some edge $e$ whose ``saturation'' (i.e., $\abs{e \cap S}$ increased from $k - \delta$ to $k - \delta/2$. Note that each edge can only increase its saturation in this way at most $\log r$ times before it contains $k$ elements in the IS. Thus, after at most $\Delta \log r$ phases, each vertex $v$ has either been added to the IS or is incident to a saturated edge. In what follows below, we describe a slight modification of this procedure that computes a $k$-MIS.

Before describing the procedure more formally, we introduce some notation. Suppose $S$ is an independent set in $H = (V, E)$ and $k$ is an integer with $1 \leq k \leq r - 1$.
\begin{itemize}
    \item The \emph{saturation} of an edge $e$ is $\sat(e) = \abs{e \cap S}$. We say $e$ is \emph{saturated} if $\sat(e) = k$. 
    \item We say that $v$ is \emph{active} if (1) $v \notin S$, and (2) $v$ is not contained in a saturated edge.
    \item We denote the set of active vertices in $e$ by $A_e$ and say that $e$ is \emph{inactive} if $A_e = \varnothing$.
    \item We say that $e$ is in \emph{phase} $\varphi$ if $(1 - 2^{-\varphi - 1})k < \sat(e) \leq (1 - 2^{-\varphi})k$. We denote the phase of $e$ by $\varphi(e)$.
\end{itemize}
We make the following observation about the phase of edges.

\begin{lemma}\label{lem:phase}
    If a edge $e$ is in phase $\varphi = \log r + 1$, then $e$ is saturated. In particular, each edge's phase can increase at most $\log r$ times before the edge is either saturated or inactive.
\end{lemma}

The main procedure assumes that the vertices in the input are properly colored (in the underlying graph) with $O(\Delta^2 r^2)$ colors. Such a coloring can be found in $O(\log^* n)$ rounds using the algorithm of Linial~\cite{linial1992locality}.
\begin{algorithm}
\caption{$k$-weak MIS\label{alg:k-weak-mis}}
\begin{algorithmic}[1]
    \Procedure{kWeakMIS}{$H$}
    \State $S \gets \varnothing$\Comment{The IS}
    \State for each edge $e$, $A_e \gets e$\Comment{Active vertices in $e$}
    \For{each iteration $i = 1, 2, \ldots, 1 + \Delta \log r$}\label{ln:lp-start}
    \If{$\abs{A_e} \geq 4 (r - k)$}\label{ln:partition-start}
    \State partition $A_e$ into $\lfloor \abs{A_e}/4 \rfloor$ parts of size at most $5$
    \Else
    \State keep $A_e$ as a single part
    \EndIf\label{ln:partition-end}
    \State $H' \gets$ the hypergraph induced by this edge partition
    \State properly color $H'$ with $4 \Delta (r - k) + 1$ colors\label{ln:color}
    \For{each color class $c$}
    \State $V_c \gets$ vertices of color $c$
    \If{$v \in V_c$ and each $e \ni v$ satisfies $\abs{V_c \cap e} \leq k - \sat(e)$}
    \State $v$ adds itself to $S$ and removes itself from $A_e$\label{ln:add-self}
    \EndIf
    \EndFor
    \If{$v$ is is incident to a saturated edge}
    \State remove $v$ from all incident $A_e$
    \EndIf
    \EndFor\label{ln:lp-end}
    \EndProcedure
\end{algorithmic}
\end{algorithm}

We now state and prove our main result for this section.

\begin{theorem}\label{thm:k-weak-mis}
    Let $H = (V, E)$ be an $r$-uniform hypergraph with maximum degree at most $\Delta$. Then on input $H$, Algorithm~\ref{alg:k-weak-mis} produces a $k$-weak MIS in $O(\Delta^2 (r - k) \log r + \Delta \log r \log^* r + \log^* n)$ rounds in the LOCAL model.
\end{theorem}
\begin{proof}
    We first establish the correctness of the procedure. To this end, let $v$ be an arbitrary vertex and $e_1, e_2, \ldots, e_d$ its incident edges. Suppose $v \notin S$ at the beginning of an iteration $i$ of the loop in lines~\ref{ln:lp-start}--\ref{ln:lp-end}. For $j = 1, 2, \ldots, d$, let $f_j \subseteq e_j$ denote the part of $e_j$ containing $v$ in the partition formed in Lines~\ref{ln:partition-start}--\ref{ln:partition-end}. By the choice of the size of the partitions, for each color $c$ in the coloring we have 
    \begin{equation}
        \label{eqn:class-size}
        \abs{V_c \cap e} \leq \frac{1}{2} (k - \sat(e))
    \end{equation}
    This expression is clearly true when $\abs{V_c} = 1$ (i.e., $A_e$ is maintained as a single part). On the other hand, if $\abs{A_e} \geq 4 (r - k)$ (hence $A_e$ is partitioned into multiple parts), then we have
    \begin{equation*}
        \abs{V_c \cap e} \leq \frac 1 4 \abs{A_e} < \frac 1 2 \cdot \frac 3 4 \abs{A_e} \leq \frac 1 2 \paren{\abs{A_e} - (r - k)} \leq \frac 1 2 (k - \sat(e)),
    \end{equation*}
    where the final inequality holds because $\abs{A_e} + \sat(e) \leq r$. 

    By~(\ref{eqn:class-size}), if $v$ does not add itself to $S$ in Line~\ref{ln:add-self}, it must be that some incident edge $e'$ increased its phase $\phi$ during this iteration of the the loop. To see this, suppose $v$ does not added to $S$ in iteration $i$ in which $v$ is colored $c$. Then $v$ has some incident edge $e$ such that $\abs{V_c \cap e} > k - \sat(e)$. However, at the beginning of iteration~(\ref{eqn:class-size}) was satisfied, hence $k - \sat(e)$ must have halved during iteration $i$, implying that its phase, $\phi(e)$ increased.
    
    Since $v$ is incident to at most $\Delta$ edges, and each incident edge can increase its phase at most $\log r$ times before becoming saturated or inactive (Lemma~\ref{lem:phase}), after $\Delta \log r$ iterations, either $v$ added itself to $S$, or it is incident to some saturated edge (hence inactive). Thus, $S$ is a $k$-weak MIS.

    For the running time analysis, we can initially color all of the vertices with $\Delta^2 r^2$ colors in $O(\log^* n)$ rounds using Linial's algorithm~\cite{linial1992locality}. To color the hypergraph $H'$ in Line~\ref{ln:color}, observe that each part of the partition has size at most $\max{5, 4(r - k)}$. Therefore, the underlying graph of $H'$ has maximum degree at most $O(\Delta(r - k)$. Given the initial coloring, the coloring algorithm of~\cite{maus2020local} has running time $O(\sqrt{\Delta (r - k)} + \log^* (r^2 \Delta^2)) = O(\Delta(r-k) + \log^*r)$. Iterating over the colors takes $O(\Delta(r - k))$ rounds, so the overall running time is $O(\Delta^2 (r - k) \log r +  \Delta \log r \log^* r + \log^* n$.
\end{proof}

\section{Conclusion}

In this paper, we have provided an overview of the $k$-MIS and $(\alpha, \beta)$-IS problems. Explicitly, we have shown lower bounds in terms of both $r$ and $\Delta$ on finding a $1$-MIS, and in terms of $\Delta$ on finding a $k$-MIS for certain values of $k$. In the other direction, we have provided a set of algorithmic results for finding both $k$-weak independent sets, and $(\alpha, \beta)$ independent sets. Our lower bound results suggest that finding $k$-weak MIS may be computationally costly, but the weaker $(\alpha, \beta)$-IS readily admits efficient algorithms for some range of parameters. Thus, in applications where some form of maximal independent set is desirable, the an $(\alpha, \beta)$-IS may offer sufficient ``maximality'' while allowing for significantly faster algorithms.

While our results provide some bounds on the LOCAL complexity of computing $k$-weak MIS and $(\alpha, \beta)$-IS, there are still significant gaps between the upper and lower bounds. Even for the case of $1$-weak MIS our lower bound $\Omega(\Delta + r + \log^* n)$ is quite far from the upper bound of $O(r \Delta + \log^* n)$. What is the LOCAL complexity of finding $1$-weak MIS?
Our algorithms show that for $k = r - o(r)$, the trivial upper bound of $r \Delta + \log^* n$ is not tight, at least for $\Delta$ much smaller than $r$. Can such improvements be found for $k = r - \Omega(r)$? 

More generally, what is the relationship between the complexities of $k$-weak MIS and $k'$-weak MIS? To this end, we observe that for $k < k'$, a $k$-weak MIS can always be extended to a $k'$-weak MIS. On the other hand, a $k'$-weak MIS need not contain any $k$-weak MIS. This seems to suggest that perhaps finding $k$-weak MIS becomes harder for smaller values of $k$. Indeed, we were only able to prove our $\Omega(\Delta)$ lower bound for $k = 1$.

In general, finding $(\alpha, \beta)$-IS seems much easier than finding $k$-weak MIS when $\beta - \alpha$ is reasonably large. Indeed, for $\alpha = 1$ and $\beta \sim c \log (r \Delta)$, $(\alpha, \beta)$-IS can be found without any communication. More generally, how does the complexity of finding an $(\alpha, \beta)$-IS vary with $\Delta$, $r$, and $\delta = \beta - \alpha + 1$?

%There are two clear directions for future research. First is investigating stronger deterministic and randomized algorithms for $k$-MIS in the general case. 
%It would seem likely that there exists some approach requiring fewer than $O(\Delta r)$ rounds in the general case. Secondly, the opposite direction of finding lower bounds remains open. This includes both finding a stronger bound for $k$-MIS than $\Omega(\Delta + r + \log^* n)$, and finding bounds for $(\alpha, \beta)$-IS when $\alpha \neq \beta$. In both cases, we conjecture that there exists some bound that is a function of both $\Delta$ and $r$, i.e. some bound of the form $O(f(\Delta) g(r))$. We further assume this bound is strictly smaller than $\Delta r$ for every value of $k$.

Our algorithms in Section~\ref{sec:deterministic-algorithms} relied upon finding $\delta$-edge defective colorings in hypergraphs---i.e., colorings in which each color class's intersection with each edge is at most $\delta$. Our approach employed known algorithms for (proper) graph coloring to find $\delta$-edge defective colorings. We think it would be interesting perform a more thorough investigation of distributed edge-defective colorings in hypergraphs. To this end, we believe the unified approach to distributed graph coloring described by Maus~\cite{maus2023-distributed} may be relevant.

\bibliography{bib.bib}

\begin{thebibliography}{10}

\bibitem{awerbuch1989network}
Baruch Awerbuch, Andrew~V Goldberg, Michael Luby, and Serge~A Plotkin.
\newblock Network decomposition and locality in distributed computation.
\newblock In {\em FOCS}, volume~30, pages 364--369. Citeseer, 1989.

\bibitem{balliu2021lower}
Alkida Balliu, Sebastian Brandt, Juho Hirvonen, Dennis Olivetti, Mika{\"e}l
  Rabie, and Jukka Suomela.
\newblock Lower bounds for maximal matchings and maximal independent sets.
\newblock {\em Journal of the ACM (JACM)}, 68(5):1--30, 2021.

\bibitem{Balliu2023-distributed}
Alkida Balliu, Sebastian Brandt, Fabian Kuhn, and Dennis Olivetti.
\newblock {\em Distributed Maximal Matching and Maximal Independent Set on
  Hypergraphs}, pages 2632--2676.
\newblock URL:
  \url{https://epubs.siam.org/doi/abs/10.1137/1.9781611977554.ch100}, \href
  {https://arxiv.org/abs/https://epubs.siam.org/doi/pdf/10.1137/1.9781611977554.ch100}
  {\path{arXiv:https://epubs.siam.org/doi/pdf/10.1137/1.9781611977554.ch100}},
  \href {https://doi.org/10.1137/1.9781611977554.ch100}
  {\path{doi:10.1137/1.9781611977554.ch100}}.

\bibitem{balliu2021improved}
Alkida Balliu, Sebastian Brandt, Fabian Kuhn, and Dennis Olivetti.
\newblock Improved distributed lower bounds for mis and bounded (out-) degree
  dominating sets in trees.
\newblock In {\em Proceedings of the 2021 ACM Symposium on Principles of
  Distributed Computing}, pages 283--293, 2021.

\bibitem{balliu2022distributed}
Alkida Balliu, Sebastian Brandt, and Dennis Olivetti.
\newblock Distributed lower bounds for ruling sets.
\newblock {\em SIAM Journal on Computing}, 51(1):70--115, 2022.

\bibitem{barenboim2009distributed}
Leonid Barenboim and Michael Elkin.
\newblock Distributed ($\delta$+ 1)-coloring in linear (in $\delta$) time.
\newblock In {\em Proceedings of the forty-first annual ACM symposium on Theory
  of computing}, pages 111--120, 2009.

\bibitem{barenboim2016locality}
Leonid Barenboim, Michael Elkin, Seth Pettie, and Johannes Schneider.
\newblock The locality of distributed symmetry breaking.
\newblock {\em Journal of the ACM (JACM)}, 63(3):1--45, 2016.

\bibitem{brandt2019automatic}
Sebastian Brandt.
\newblock An automatic speedup theorem for distributed problems.
\newblock In {\em Proceedings of the 2019 ACM Symposium on Principles of
  Distributed Computing}, pages 379--388, 2019.

\bibitem{brandt2016lower}
Sebastian Brandt, Orr Fischer, Juho Hirvonen, Barbara Keller, Tuomo
  Lempi{\"a}inen, Joel Rybicki, Jukka Suomela, and Jara Uitto.
\newblock A lower bound for the distributed lov{\'a}sz local lemma.
\newblock In {\em Proceedings of the forty-eighth annual ACM symposium on
  Theory of Computing}, pages 479--488, 2016.

\bibitem{erdos1975problems}
Paul Erdos and L{\'a}szl{\'o} Lov{\'a}sz.
\newblock Problems and results on 3-chromatic hypergraphs and some related
  questions.
\newblock {\em Infinite and finite sets}, 10(2):609--627, 1975.

\bibitem{ghaffari2016improved}
Mohsen Ghaffari.
\newblock An improved distributed algorithm for maximal independent set.
\newblock In {\em Proceedings of the twenty-seventh annual ACM-SIAM symposium
  on Discrete algorithms}, pages 270--277. SIAM, 2016.

\bibitem{ghaffari2024nearoptimaldeterministicnetworkdecomposition}
Mohsen Ghaffari and Christoph Grunau.
\newblock Near-optimal deterministic network decomposition and ruling set, and
  improved mis, 2024.
\newblock URL: \url{https://arxiv.org/abs/2410.19516}, \href
  {https://arxiv.org/abs/2410.19516} {\path{arXiv:2410.19516}}.

\bibitem{ghaffari2021improved}
Mohsen Ghaffari, Christoph Grunau, and V{\'a}clav Rozho{\v{n}}.
\newblock Improved deterministic network decomposition.
\newblock In {\em Proceedings of the 2021 ACM-SIAM Symposium on Discrete
  Algorithms (SODA)}, pages 2904--2923. SIAM, 2021.

\bibitem{kuhn2018efficient}
Fabian Kuhn and Chaodong Zheng.
\newblock Efficient distributed computation of mis and generalized mis in
  linear hypergraphs.
\newblock {\em arXiv preprint arXiv:1805.03357}, 2018.

\bibitem{kutten2014distributed}
Shay Kutten, Danupon Nanongkai, Gopal Pandurangan, and Peter Robinson.
\newblock Distributed symmetry breaking in hypergraphs.
\newblock In {\em Distributed Computing: 28th International Symposium, DISC
  2014, Austin, TX, USA, October 12-15, 2014. Proceedings 28}, pages 469--483.
  Springer, 2014.

\bibitem{linial1987distributive}
Nathan Linial.
\newblock Distributive graph algorithms global solutions from local data.
\newblock In {\em 28th Annual Symposium on Foundations of Computer Science
  (sfcs 1987)}, pages 331--335. IEEE, 1987.

\bibitem{linial1992locality}
Nathan Linial.
\newblock Locality in distributed graph algorithms.
\newblock {\em SIAM Journal on computing}, 21(1):193--201, 1992.

\bibitem{maus2023-distributed}
Yannic Maus.
\newblock Distributed graph coloring made easy.
\newblock {\em ACM Transactions on Parallel Computing}, 10(4):1–21, December
  2023.
\newblock URL: \url{http://dx.doi.org/10.1145/3605896}, \href
  {https://doi.org/10.1145/3605896} {\path{doi:10.1145/3605896}}.

\bibitem{maus2020local}
Yannic Maus and Tigran Tonoyan.
\newblock Local conflict coloring revisited: Linial for lists.
\newblock {\em arXiv preprint arXiv:2007.15251}, 2020.

\bibitem{Mitzenmacher_Upfal_2005}
Michael Mitzenmacher and Eli Upfal.
\newblock {\em Chernoff Bounds}, page 61–89.
\newblock Cambridge University Press, 2005.

\bibitem{Moser2010-constructive}
Robin~A. Moser and Gábor Tardos.
\newblock A constructive proof of the general lovász local lemma.
\newblock {\em Journal of the ACM}, 57(2):1–15, January 2010.
\newblock URL: \url{http://dx.doi.org/10.1145/1667053.1667060}, \href
  {https://doi.org/10.1145/1667053.1667060}
  {\path{doi:10.1145/1667053.1667060}}.

\bibitem{panconesi2001some}
Alessandro Panconesi and Romeo Rizzi.
\newblock Some simple distributed algorithms for sparse networks.
\newblock {\em Distributed computing}, 14(2):97--100, 2001.

\end{thebibliography}

\end{document}